\documentclass[a4paper,UKenglish,cleveref]{lipics-v2019}
\title{Colouring $(sP_1+P_5)$-Free Graphs: a Mim-Width Perspective}

\titlerunning{Colouring $(sP_1+P_5)$-free graphs: a mim-width perspective}

\author{Nick Brettell}{Department of Computer Science, Durham University, UK}{nbrettell@gmail.com}{https://orcid.org/0000-0002-1136-418X}{}

\author{Jake Horsfield}{School of Computing, University of Leeds, Leeds,
UK}{sc15jh@leeds.ac.uk}{https://orcid.org/0000-0002-4388-5123}{}

\author{Dani\"el Paulusma}{Department of Computer Science, Durham University, UK}{daniel.paulusma@durham.ac.uk}{https://orcid.org/0000-0001-5945-9287}{}

\authorrunning{N. Brettell, J. Horsfield, and D. Paulusma}

\Copyright{Nick Brettell, Jake Horsfield, and Dani\"el Paulusma}



\funding{The research in this paper received support from the Leverhulme Trust (RPG-2016-258).}

\acknowledgements{}

\nolinenumbers

\hideLIPIcs

\usepackage{amsmath,amssymb,amsthm,xspace,xcolor,tikz,graphicx,graphics,color,tkz-graph}
\usetikzlibrary{fit,automata,arrows}

\newcommand{\NP}{{\sf NP}}

\newcommand{\W}{{\sf W}}
\newcommand{\cutmim}{\mathrm{cutmim}}
\newcommand{\mimw}{\mathrm{mimw}}

\begin{document}

\maketitle

\begin{abstract} 
We prove that the class of $(K_t,sP_1+P_5)$-free graphs has bounded mim-width for every $s\geq 0$ and $t\geq 1$, and that there is a polynomial-time algorithm that, given a graph in the class, computes a branch decomposition of constant mim-width. A large number of \NP-complete graph problems become polynomial-time solvable on graph classes with bounded mim-width and for which a branch decomposition is quickly computable. The $k$-{\sc Colouring} problem is an example of such a problem. For this problem, we may assume that the input graph is $K_{k+1}$-free. Then, as a consequence of our result, we obtain a new proof for the known result that for every fixed $k\geq 1$ and $s\geq 0$, $k$-{\sc Colouring} is polynomial-time solvable for $(sP_1+P_5)$-free graphs.  
In fact, our findings show that the underlying reason for this polynomial-time algorithm is that the class {\it has bounded mim-width}.
\end{abstract}

\section{Introduction}

A graph class ${\cal G}$ has bounded ``width'' if there exists a constant $c$ such that every graph in ${\cal G}$ has width at most~$c$.
We continue a recent study~\cite{BMPP} on the boundedness of a particular width parameter, namely {\it mim-width}, for hereditary graph classes. 
Mim-width was introduced by Vatshelle~\cite{Va12}, who proved that it is more powerful than clique-width, and several well-known width parameters of equivalent strength, including boolean width, module-width, NLC-width and rank-width. That is, every graph class of bounded clique-width has bounded mim-width, but there exist graph classes of bounded mim-width that have unbounded clique-width. Hence, proving that a problem is polynomial-time solvable for graphs of bounded mim-width yields larger ``islands of tractability'' than doing this for clique-width (or any parameter equivalent to clique-width) in terms of the graphs that admit a polynomial-time algorithm. However, fewer problems have such an algorithm; see~\cite{BV13,BK19,BPT19,BTV13,GMR20,JKST19,JKT,JKT19} for some examples.

To define mim-width, we need the notion of a \textit{branch decomposition} for a graph~$G$, which is a pair $(T, \delta)$, where $T$ is a subcubic tree and $\delta$ is a bijection from~$V(G)$ to the leaves of $T$.  
We note that each edge $e \in E(T)$ partitions the leaves of $T$ into two classes, $L_e$ and $\overline{L_e}$, depending on which component of $T-e$ they belong to.
In this way, each edge $e$ induces a partition $(A_e, \overline{A_e})$ of $V(G)$, where $\delta(A_e) = L_e$ and $\delta(\overline{A_e}) = \overline{L_e}$.
Let $G[A_e,\overline{A_e}]$ denote the bipartite subgraph of $G$ induced by the edges with one end-vertex in $A_e$ and the other in $\overline{A_e}$.
A matching $F \subseteq E(G)$ of $G$ is {\it induced} if there is no edge in $G$ between vertices of different edges of $F$.
We let $\cutmim_{G}(A_{e}, \overline{A_{e}})$ denote the size of a maximum induced matching in $G[A_{e}, \overline{A_{e}}]$.
The \emph{mim-width} $\mimw_{G}(T, \delta)$ of $(T, \delta)$ is the maximum value of $\cutmim_{G}(A_{e}, \overline{A_{e}})$ over all edges $e\in E(T)$. The \emph{mim-width} $\mimw(G)$ of $G$ is the minimum value of $\mimw_{G}(T, \delta)$ over all branch decompositions $(T, \delta)$ for $G$. 

S{\ae}ther and Vatshelle~\cite{SV16} proved that computing mim-width is \NP-hard. In the same paper they  showed that deciding if the mim-width is at most $k$ is \W$[1]$-hard when parameterized by $k$, and that there is no polynomial-time algorithm for approximating the mim-width of a graph to within a constant factor of the optimal, unless $\mathsf{NP} = \mathsf{ZPP}$.
However, there are many graph classes (see~\cite{BMPP,BTV13,KKST17}) for which mim-width is bounded and {\it quickly computable}.
By the latter, we mean that the class admits a polynomial-time algorithm for computing a branch decomposition whose mim-width is bounded by a constant.

A class of graphs is {\it hereditary} if it is closed under vertex deletion, or equivalently, can be characterized by a set of forbidden induced subgraphs.
In particular, a graph $G$ is {\it $H$-free} for some graph $H$ if $G$ does not contain $H$ as an induced subgraph, and $G$ is {\it $(H_1,H_2)$-free} for some graphs $H_1$ and $H_2$ if $G$ is both $H_1$-free and $H_2$-free.
It is not difficult to prove that a class of $H$-free graphs has bounded mim-width if and only if $H$ is an induced subgraph of the $4$-vertex path $P_4$; see~\cite{BMPP}, in which a study into the boundedness of mim-width of $(H_1,H_2)$-free graphs was initialized. We also refer to~\cite{BMPP} for an up-to-date summary and list of open cases. 
In this brief note, we focus on a specific family of $(H_1,H_2)$-free graphs. 
The \textit{disjoint union} $G+H$ of graphs $G$ and $H$ is the graph with vertex set $V(G) \cup V(H)$ and edge set $E(G) \cup E(H)$. The graph $kG$ denotes the disjoint union of $k$ copies of $G$. The path and complete graph on $n$ vertices are denoted by $P_n$ and $K_{n}$, respectively. We can now state our main result.

\begin{theorem}\label{t-main}
For every $s\geq 0$ and $t\geq 1$, the mim-width of the class of $(K_t,sP_1+P_5)$-free graphs is bounded 
and quickly computable.
\end{theorem}

\noindent
Apart from solving an infinite family of open cases, we believe this result is of interest as it sheds some light on a known result on graph colouring. 
A {\em colouring} of a graph $G=(V,E)$ is a mapping $c: V\rightarrow\{1,2,\ldots \}$ that assigns each vertex~$u\in V$ a {\it colour} $c(u)$ in such a way that $c(u)\neq c(v)$ whenever $u$ and $v$ are adjacent.
If $1\leq c(u)\leq k$ for each $u \in V$, then $c$ is also called a {\it
$k$-colouring} of $G$. The {\sc Colouring} problem is to decide, on
input a graph $G$ and an integer $k$, whether $G$ has a $k$-colouring. For fixed $k$ (that is, $k$ is not part of the input), the problem is known as $k$-{\sc Colouring}. 

It is well known that {\sc $3$-Colouring} is \NP-complete~\cite{Lo73},
but the problem becomes polynomial-time solvable under certain input restrictions. This led to a large number of complexity results for {\sc Colouring} and $k$-{\sc Colouring} on special graph classes.
The complexity of {\sc Colouring} for $H$-free graphs has been settled~\cite{KKTW01}, but there are still 
a number of open cases for {\sc $k$-Colouring} restricted to $H$-free graphs when $H$ is a {\it linear forest}, that is, 
a disjoint union of paths; see~\cite{GJPS17} for a survey and~\cite{CHSZ18,KMMNPS18} for some later summaries. 
In particular, Ho\`ang et al.~\cite{HKLSS10} proved that for every integer $k\geq 1$, {\sc $k$-Colouring} is polynomial-time solvable for $P_5$-free graphs. This result was generalized by Couturier et al.~\cite{CGKP15}.

\begin{theorem}[\cite{CGKP15}]\label{t-known}
For every $k\geq 1$ and $s\geq 0$, $k$-{\sc Colouring} is polynomial-time solvable for $(sP_1+P_5)$-free graphs. 
\end{theorem}

We note that Theorem~\ref{t-main} implies Theorem~\ref{t-known}, as $k$-{\sc Colouring}, for every fixed integer $k\geq 1$, is polynomial-time solvable for a graph class where mim-width is bounded and quickly computable\footnote{This is not true for the {\sc Colouring} problem, where $k$ is part of the input; indeed {\sc Colouring} is \NP-complete for circular-arc graphs~\cite{GJMP80}, despite the fact that this is a class for which mim-width is bounded and quickly computable~\cite{BV13}.}~\cite{BTV13}, and we may assume that an instance of $k$-{\sc Colouring} is $K_{k+1}$-free, for otherwise it is clearly not $k$-colourable.

For an integer $k\geq 1$, a  {\it $k$-list assignment} of a graph
$G=(V,E)$ is a function $L$ that assigns each vertex $u\in V$ a {\it list} $L(u)\subseteq \{1,2,\ldots,k\}$ of {\it admissible} colours for $u$. A colouring $c$ of $G$ {\it respects} $L$ if  $c(u)\in L(u)$ for every $u\in V$. 
For fixed $k\geq 1$, the {\sc List $k$-Colouring} problem is to decide if a graph $G$ has a colouring that respects a $k$-list assignment $L$.
Theorem~\ref{t-known} holds even for {\sc List $k$-Colouring}. 
It is known that {\sc List $k$-Colouring} is polynomial-time solvable for graph classes of bounded clique-width~\cite{KR03}.
Kwon~\cite{Kw20} observed the following. Given an instance ($G,L)$ of {\sc List $k$-Colouring}, one can construct an equivalent instance $G'$ of {\sc $k$-Colouring} by adding a clique on new vertices $u_1,\ldots,u_k$ to $G$ and adding an edge between $u_i$ and $v\in V(G)$ if and only if $i\notin L(u)$. As
$\mimw(G')\leq \mimw(G)+k$, this means that {\sc List $k$-Colouring} is polynomial-time solvable even for graph classes of bounded mim-width. Hence, by the same arguments as before,  Theorem~\ref{t-main} also serves as an alternative proof for the aforementioned {\sc List $k$-Colouring} generalization of Theorem~\ref{t-known}.

Let $\omega(G)$ denote the size of a largest clique in a graph $G$. Note that $G$ is $K_t$-free if and only if $\omega(G)\leq t-1$.
Very recently, Chudnovsky et al.~\cite{CKPR20} independently gave for the class of $P_5$-free graphs, an $n^{O(\omega(G))}$-time algorithm for \textsc{Max Partial $H$-Colouring}, which is equivalent to {\sc Independent Set} if $H=P_1$ and to {\sc Odd Cycle Transversal} if $H=P_2$. 
They noted that \textsc{Max Partial $H$-Colouring} is polynomial-time solvable for graph classes whose mim-width is bounded and quickly computable. Hence, Theorem~\ref{t-main} also provides a polynomial-time algorithm for \textsc{Max Partial $H$-Colouring} restricted to 
$(K_t,sP_1+P_5)$-free graphs for every $s\geq 0$ and $t\geq 1$, or equivalently, $(sP_1+P_5)$-free graphs $G$ with $\omega(G)\leq t-1$. However, the running time of this algorithm is worse than $n^{O(\omega(G))}$ (see~\cite{CKPR20} for details).
 
The rest of the paper is devoted to the proof of Theorem~\ref{t-main}.

\section{The Proof of Theorem~\ref{t-main}}

For a graph $G=(V,E)$,  a set $D\subseteq V$ is {\it dominating} if every vertex of $G$ belongs to~$D$ or is adjacent to a vertex of $D$; we also say that  $G[D]$ is {\em dominating}.
We need the following two known results on $(sP_1+P_5)$-free graphs, which distinguish the cases $s=0$ and $s>0$.

\begin{lemma}[\cite{BT90}]\label{l-bt}
Every connected $P_5$-free graph $G$
has a dominating $P_3$ or a dominating complete graph. 
\end{lemma}

\begin{lemma}[\cite{CGKP15}]\label{l-ds}
For an integer $s\geq 1$, let $G$ be an $(sP_1+P_5)$-free graph.
If $G$ contains an induced $P_5$, then $G$ contains a dominating induced 
$rP_1+P_5$ for some $r<s$.
\end{lemma}

Let $G_1$ and $G_2$ be two vertex-disjoint graphs. We {\it identify} two vertices $u\in V(G_1)$ and $v\in V(G_2)$ by adding an edge between them, which we then contract.

We will need a new result on mim-width, which might be of independent interest.
It shows that, given a partition of the vertex set of a graph $G$, we can bound the mim-width of $G$ in terms of the mim-width of the graphs induced by each part and the mim-width between any two of the parts.

\begin{lemma}
\label{mimmultijoin}
Let $G$ be a graph, and let $(X_1,\dotsc,X_p)$ be a partition of $V(G)$ such that $\cutmim_G(X_i,X_j) \le c$ for all distinct $i,j \in \{1,\dotsc,p\}$, and $p \ge 2$.  Then \[\mimw(G) \le \max\left\{c\left\lfloor\left(\frac{p}{2}\right)^2\right\rfloor,\max_{i \in \{1,\dotsc,p\}}\{\mimw(G[X_i])\} + c(p-1)\right\}.\]
Moreover, if $(T_i,\delta_i)$ is a branch decomposition of $G[X_i]$ for each $i$, then we can construct, in $O(1)$ time, a branch decomposition $(T,\delta)$ of $G$ with $\mimw(T,\delta) \le \max\{c\lfloor(\frac{p}{2})^2\rfloor,\max_{i \in \{1,\dotsc,p\}}\{\mimw(T_i,\delta_i)\} + c(p-1)\}$.
\end{lemma}
\begin{proof}
We construct a branch decomposition $(T,\delta)$ of $G$ with the desired mim-width, as follows.
Let $T_0$ be an arbitrary subcubic tree having $p$ leaves $\ell_1,\dotsc,\ell_p$.
For each $i \in \{1,\dotsc,p\}$, we choose an arbitrary leaf vertex $v_i$ of $T_i$, we identify $v_i$ with $\ell_i$ calling the resulting vertex $\ell_i$, and we create a new pendant edge incident to $\ell_i$, where the new leaf vertex adjacent to $\ell_i$ is called $v_i$.
Then $T$ is a subcubic tree whose set of leaves is the disjoint union of the leaves of $T_i$ for each $i \in \{1,\dotsc,p\}$. 
See \cref{trees-fig}, for example.
For a leaf $v$ of $T$, we set $\delta(v) = \delta_i(v)$ where $v$ is a leaf of $T_i$.
Now $(T,\delta)$ is a branch decomposition of $G$, and clearly this branch decomposition can be constructed in $O(1)$ time.
It remains to prove the upper bound for $\mimw(T,\delta)$.

\begin{figure}[hb]
  \centering
  \begin{tikzpicture}[scale=0.39,line width=1pt]
    \tikzset{VertexStyle/.append style = {minimum height=5,minimum width=5}}
    \node at (1,-1.5) {\large$T_0$};

    \Vertex[x=-1,y=1,LabelOut=true,L=$\ell_1$,Lpos=180]{l1}
    \Vertex[x=-1,y=-1,LabelOut=true,L=$\ell_2$,Lpos=180]{l2}
    \Vertex[x=3,y=1,LabelOut=true,L=$\ell_{p-1}$,Lpos=0]{l3}
    \Vertex[x=3,y=-1,LabelOut=true,L=$\ell_{p}$,Lpos=0]{l4}

    \SetVertexNoLabel
    \tikzset{VertexStyle/.append style = {fill=gray}}
    \Vertex[x=0,y=0]{i0}
    \Vertex[x=2,y=0]{i1}

    \Edge(i0)(l1)
    \Edge(i0)(l2)
    \tikzset{EdgeStyle/.append style = {dotted}}
    \Edge(i0)(i1)
    \tikzset{EdgeStyle/.append style = {solid}}
    \Edge(i1)(l3)
    \Edge(i1)(l4)

    \node at (7.0,2) {\large$T_1$};

    \tikzset{VertexStyle/.append style = {fill=black}}
    \Vertex[x=8,y=3]{t1_1}
    \Vertex[x=8,y=1]{t1_2}
    \Vertex[x=12,y=3]{t1_3}
    \SetVertexLabel
    \Vertex[x=12,y=1,LabelOut=true,L=$v_1$,Lpos=0]{t1_4}
    \SetVertexNoLabel

    \tikzset{VertexStyle/.append style = {fill=gray}}
    \Vertex[x=9,y=2]{t1_i0}
    \Vertex[x=11,y=2]{t1_i1}

    \Edge(t1_i0)(t1_1)
    \Edge(t1_i0)(t1_2)
    \Edge(t1_i0)(t1_i1)
    \Edge(t1_i1)(t1_3)
    \Edge(t1_i1)(t1_4)

    \node at (7.0,-2) {\large$T_2$};

    \tikzset{VertexStyle/.append style = {fill=black}}
    \Vertex[x=8,y=-3]{t2_1}
    \Vertex[x=8,y=-1]{t2_2}
    \Vertex[x=12,y=-3]{t2_3}
    \Vertex[x=10,y=-3.3]{t2_5}
    \SetVertexLabel
    \Vertex[x=12,y=-1,LabelOut=true,L=$v_2$,Lpos=0]{t2_4}
    \SetVertexNoLabel

    \tikzset{VertexStyle/.append style = {fill=gray}}
    \Vertex[x=9,y=-2]{t2_i0}
    \Vertex[x=10,y=-2]{t2_i2}
    \Vertex[x=11,y=-2]{t2_i1}

    \Edge(t2_i0)(t2_1)
    \Edge(t2_i0)(t2_2)
    \Edge(t2_i0)(t2_i2)
    \Edge(t2_i2)(t2_i1)
    \Edge(t2_i2)(t2_5)
    \Edge(t2_i1)(t2_3)
    \Edge(t2_i1)(t2_4)

    \node at (20.0,2.5) {\large$T_{p-1}$};

    \tikzset{VertexStyle/.append style = {fill=black}}
    \Vertex[x=18,y=3.3]{t3_1}
    \SetVertexLabel
    \Vertex[x=17,y=1,LabelOut=true,L=$v_{p-1}$,Lpos=180]{t3_2}
    \SetVertexNoLabel
    \Vertex[x=19,y=1]{t3_4}

    \tikzset{VertexStyle/.append style = {fill=gray}}
    \Vertex[x=18,y=2]{t3_i0}

    \Edge(t3_i0)(t3_1)
    \Edge(t3_i0)(t3_2)
    \Edge(t3_i0)(t3_4)

    \node at (21.0,-2) {\large$T_p$};

    \tikzset{VertexStyle/.append style = {fill=black}}
    \Vertex[x=16,y=-3]{t4_1}
    \SetVertexLabel
    \Vertex[x=16,y=-1,LabelOut=true,L=$v_p$,Lpos=180]{t4_2}
    \SetVertexNoLabel
    \Vertex[x=20,y=-3]{t4_3}
    \Vertex[x=20,y=-1]{t4_4}
    \Vertex[x=17,y=-4.3]{t4_5}
    \Vertex[x=19,y=-4.3]{t4_6}

    \tikzset{VertexStyle/.append style = {fill=gray}}
    \Vertex[x=17,y=-2]{t4_i0}
    \Vertex[x=19,y=-2]{t4_i1}
    \Vertex[x=18,y=-2]{t4_i2}
    \Vertex[x=18,y=-3.3]{t4_i3}

    \Edge(t4_i0)(t4_1)
    \Edge(t4_i0)(t4_2)
    \Edge(t4_i0)(t4_i2)
    \Edge(t4_i2)(t4_i1)
    \Edge(t4_i1)(t4_3)
    \Edge(t4_i1)(t4_4)
    \Edge(t4_i2)(t4_i3)
    \Edge(t4_i3)(t4_5)
    \Edge(t4_i3)(t4_6)

    \node at (28.0,-5) {\large$T$};

    \SetVertexLabel
    \tikzset{VertexStyle/.append style = {fill=white}}
    \Vertex[x=28,y=1,LabelOut=true,L=$\ell_1$,Lpos=180]{t1}
    \Vertex[x=28,y=-1,LabelOut=true,L=$\ell_2$,Lpos=180]{t2}
    \Vertex[x=34,y=1.5,LabelOut=true,L=$\ell_{p-1}$,Lpos=-90]{t3}
    \Vertex[x=33,y=-1,LabelOut=true,L=$\ell_{p}$,Lpos=0]{t4}

    \SetVertexNoLabel
    \tikzset{VertexStyle/.append style = {fill=gray}}
    \Vertex[x=29,y=0]{i0}
    \Vertex[x=32,y=0]{i1}

    \Edge(i0)(t1)
    \Edge(i0)(t2)
    \tikzset{EdgeStyle/.append style = {dotted}}
    \Edge(i0)(i1)
    \tikzset{EdgeStyle/.append style = {solid}}
    \Edge(i1)(t3)
    \Edge(i1)(t4)

    \tikzset{VertexStyle/.append style = {fill=black}}
    \Vertex[x=24,y=3]{t_1_1}
    \Vertex[x=24,y=1]{t_1_2}
    \Vertex[x=28,y=3]{t_1_3}
    \SetVertexLabel
    \Vertex[x=29,y=2,LabelOut=true,L=$v_1$,Lpos=0]{t_1_4}
    \SetVertexNoLabel

    \tikzset{VertexStyle/.append style = {fill=gray}}
    \Vertex[x=25,y=2]{t_1_i0}
    \Vertex[x=27,y=2]{t_1_i1}

    \Edge(t_1_i0)(t_1_1)
    \Edge(t_1_i0)(t_1_2)
    \Edge(t_1_i0)(t_1_i1)
    \Edge(t_1_i1)(t_1_3)
    \Edge(t_1_i1)(t1)
    \Edge(t1)(t_1_4)

    \tikzset{VertexStyle/.append style = {fill=black}}
    \Vertex[x=24,y=-3]{t_2_1}
    \Vertex[x=24,y=-1]{t_2_2}
    \Vertex[x=28,y=-3]{t_2_3}
    \Vertex[x=26,y=-3.3]{t_2_5}
    \SetVertexLabel
    \Vertex[x=29,y=-2,LabelOut=true,L=$v_2$,Lpos=-45]{t_2_4}
    \SetVertexNoLabel

    \tikzset{VertexStyle/.append style = {fill=gray}}
    \Vertex[x=25,y=-2]{t_2_i0}
    \Vertex[x=26,y=-2]{t_2_i2}
    \Vertex[x=27,y=-2]{t_2_i1}

    \Edge(t_2_i0)(t_2_1)
    \Edge(t_2_i0)(t_2_2)
    \Edge(t_2_i0)(t_2_i2)
    \Edge(t_2_i2)(t_2_i1)
    \Edge(t_2_i2)(t_2_5)
    \Edge(t_2_i1)(t_2_3)
    \Edge(t_2_i1)(t2)
    \Edge(t2)(t_2_4)

    \tikzset{VertexStyle/.append style = {fill=black}}
    \Vertex[x=35,y=3.8]{t_3_1}
    \SetVertexLabel
    \Vertex[x=33,y=2.5,LabelOut=true,L=$v_{p-1}$,Lpos=135]{t_3_2}
    \SetVertexNoLabel
    \Vertex[x=36,y=1.5]{t_3_4}

    \tikzset{VertexStyle/.append style = {fill=gray}}
    \Vertex[x=35,y=2.5]{t_3_i0}

    \Edge(t_3_i0)(t_3_1)
    \Edge(t_3_i0)(t3)
    \Edge(t3)(t_3_2)
    \Edge(t_3_i0)(t_3_4)

    \tikzset{VertexStyle/.append style = {fill=black}}
    \Vertex[x=33,y=-3]{t_4_1}
    \SetVertexLabel
    \Vertex[x=32,y=-2,LabelOut=true,L=$v_p$,Lpos=180]{t_4_2}
    \SetVertexNoLabel
    \Vertex[x=37,y=-3]{t_4_3}
    \Vertex[x=37,y=-1]{t_4_4}
    \Vertex[x=34,y=-4.3]{t_4_5}
    \Vertex[x=36,y=-4.3]{t_4_6}

    \tikzset{VertexStyle/.append style = {fill=gray}}
    \Vertex[x=34,y=-2]{t_4_i0}
    \Vertex[x=36,y=-2]{t_4_i1}
    \Vertex[x=35,y=-2]{t_4_i2}
    \Vertex[x=35,y=-3.3]{t_4_i3}

    \Edge(t_4_i0)(t_4_1)
    \Edge(t_4_i0)(t4)
    \Edge(t4)(t_4_2)
    \Edge(t_4_i0)(t_4_i2)
    \Edge(t_4_i2)(t_4_i1)
    \Edge(t_4_i1)(t_4_3)
    \Edge(t_4_i1)(t_4_4)
    \Edge(t_4_i2)(t_4_i3)
    \Edge(t_4_i3)(t_4_5)
    \Edge(t_4_i3)(t_4_6)
  \end{tikzpicture}
  \caption{An example of the construction of $T$ in the proof of \cref{mimmultijoin}.}
  \label{trees-fig}
\end{figure}
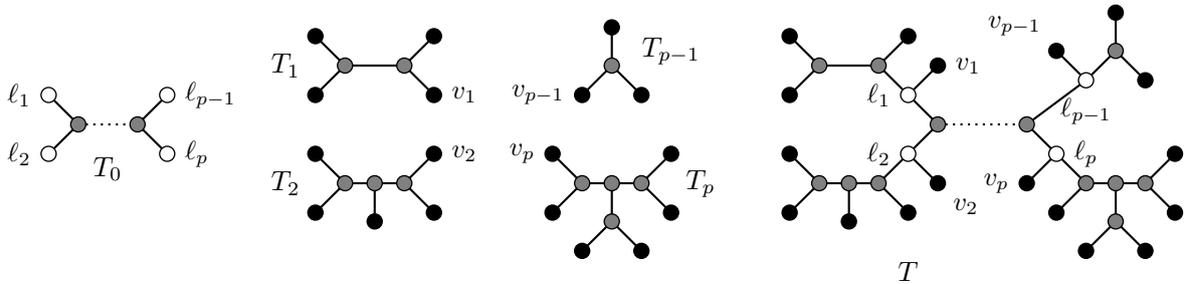

Consider $e \in E(T)$ and the partition $(A_e,\overline{A_e})$ of $V(G)$.
If $e \in E(T_0)$, then $A_e = \bigcup_{j \in J} X_j$ for some $J \subseteq \{1,\dots,p\}$.
If $e \in E(T_i)$ for some $i \in \{1,\dotsc,p\}$, then either $A_e$ or $\overline{A_e}$ is properly contained in $X_i$.
The only other possibility is that $e$ is one of the newly created pendant edges, in which case either $A_e$ or $\overline{A_e}$ has size~$1$.

First suppose $e\in E(T_0)$, so $A_e = \bigcup_{j \in J} X_j$ for some $J \subseteq \{1,\dots,p\}$.
We claim that $\cutmim_G(A_e,\overline{A_e}) \le c\left\lfloor\left(\frac{p}{2}\right)^2\right\rfloor$.
Let $M$ be a maximum-sized induced matching in $G[A_e,\overline{A_e}]$.
Let $K = \{1,\dotsc,p\} \setminus J$.
For each $j \in J$ and $k \in K$, there are at most $c$ edges of $M$ with one end in $X_j$ and the other end in $X_k$, since $\cutmim_G(X_j,X_k) \le c$.
Thus $\cutmim_G(A_e,\overline{A_e}) \le c|J||K|$, where $|J| + |K| = p$.
As $c|J||K| \le c\left\lfloor\left(\frac{p}{2}\right)\right\rfloor\left\lceil\left(\frac{p}{2}\right)\right\rceil = c\left\lfloor\left(\frac{p}{2}\right)^2\right\rfloor$, the claim follows.

Now suppose $e \in E(T_i)$ for some $i \in \{1,\dotsc,p\}$, so, without loss of generality, $A_e$ is properly contained in $X_i$.
We claim that $\cutmim_G(A_e,\overline{A_e}) \le \mimw(G[X_i]) + c(p-1)$.
Consider a maximum-sized induced matching $M$ in $G[A_e,\overline{A_e}]$.
As $A_e \subseteq X_i$, all the edges of $M$ have one end in $X_i$.
For each $j \in \{1,\dotsc,p\}$ with $j \neq i$, there are at most $c$ edges of $M$ with one end in $X_j$, since $\cutmim_G(X_i,X_j) \le c$.
Since there are at most $\mimw(G[X_i])$ edges of $M$ with both ends in $X_i$, we deduce that $\cutmim_G(A_e,\overline{A_e}) \le \mimw(G[X_i]) + c(p-1)$, as claimed.  The lemma follows.
\end{proof}

We are now ready to prove Theorem~\ref{t-main}, which we restate below.

\medskip
\noindent
{\bf Theorem~\ref{t-main} (restated).}
{\it For every $s\geq 0$ and $t\geq 1$, the mim-width of the class of $(K_t,sP_1+P_5)$-free graphs is bounded 
and quickly computable.}

\begin{proof}
Let $G=(V,E)$ be a $(K_t,sP_1+P_5)$-free graph for some $s\geq 0$ and $t\geq 1$.
We may assume without loss of generality that $G$ is connected.
We use induction on $t$. If $t=1$, then the statement of the theorem holds trivially.

Suppose that $t\geq 2$. 
First suppose that $G$ is $P_5$-free. Then, by Lemma~\ref{l-bt}, $G$ has a dominating set of size at most $\max\{3,t-1\}$.
Now suppose that $G$ is not $P_5$-free. Then, by Lemma~\ref{l-ds}, $G$ has a dominating set $D$ of size at most 
$s-1+5=s+4$.
Hence, in both cases, $G$ has a dominating set $D$ of constant size, as $s$ and $t$ are fixed; we let $p=|D|$. Moreover, 
we can find $D$ in polynomial time by brute force.

We will partition the vertex set of $G$ with respect to $D$ in the same way as in~\cite{CGKP15,HKLSS10}.
That is, we first assign an arbitrary ordering $d_1,\ldots,d_p$ on the vertices of $D$.
We then let $X_1$ be the set of vertices in $V\setminus D$ adjacent to $d_1$, 
and for $i=2,\ldots,p$, let $X_i$ be the set of vertices in $V\setminus D$
adjacent to $d_i$, but non-adjacent to any $d_h$ with $h\leq i-1$. 
Note that $D$ and the sets $X_1,\ldots,X_p$ partition $V$ (but some of the sets $X_i$ might be empty).

As $\{d_i\}$ dominates $X_i$ for every $i\in \{1,\ldots,p\}$, each $X_i$ induces a $(K_{t-1},sP_1+P_5)$-free subgraph of $G$, which we denote by $G_i$. By the induction hypothesis, the mim-width of $G_i$ is bounded and quickly computable for every $i\in \{1,\ldots,p\}$. 

Consider two sets $X_i$ and $X_j$ with $i<j$. By Ramsey's Theorem, 
for every two positive integers $p$ and~$q$, there exists an integer $R(p,q)$ such that if $G$ is a graph on at least $R(p,q)$ vertices, then $G$ has a clique of size~$p$ or an independent set of size~$q$.
We claim that $\cutmim_G(X_i,X_j) < c=R(t-1,R(t-1,s+2))$.
Towards a contradiction, suppose that $\cutmim_G(X_i,X_j) \geq c$. 
Let $A= \{a_1,a_2,\dotsc,a_c\} \subseteq X_i$ and $B =
\{b_1,b_2,\dotsc,b_c\} \subseteq X_j$ such that $a_ib_i$ is a distinct
edge of a corresponding induced matching. Then, as $G[X_i]$ is
$K_{t-1}$-free, the subgraph of $G$ induced by $A$ contains an
independent set $A'$ of size $c'=R(t-1,s+2)$. Assume without loss of generality that $A'=\{a_1,\ldots,a_{c'}\}$. Let $B'=\{b_1,\ldots,b_{c'}\}$. As $G[X_j]$ is $K_{t-1}$-free, the subgraph of $G$ induced by $B'$ contains an independent set $B''$ of size $s+2$. Assume without loss of generality that $B''=\{b_1,\ldots,b_{s+2}\}$. Then $\{b_1,a_1,d_i,a_2,b_2\} \cup \{b_3,\ldots,b_{s+2}\}$ induces an $sP_1+P_5$, a contradiction. 
We deduce that $\cutmim_G(X_i,X_j) < c$.

We now apply Lemma~\ref{mimmultijoin} to find that the mim-width of $G-D$ is bounded and quickly computable.
Let $(T,\delta)$ be a branch decomposition of $G-D$ having mim-width~$k$. 
As $|D|=p$, we can readily extend $(T,\delta)$ to a branch decomposition $(T^*,\delta^*)$ of mim-width at most $k+p$, where $T^*$ is obtained by identifying one leaf of $T$ with a leaf of an arbitrary subcubic tree having $|D|+2$ leaves.
\end{proof}

\end{document}